\documentclass[12pt]{article}

\usepackage{amsmath}
\usepackage{amssymb}
\usepackage{dsfont} 

\usepackage{graphicx}

\usepackage[dvipsnames]{xcolor}
\usepackage{color}

\newcommand{\be}{\begin{equation}}
\newcommand{\bel}[1]{\begin{equation}\label{#1}}
\newcommand{\ee}{\end{equation}}
\newcommand{\bea}{\begin{eqnarray}}
\newcommand{\eea}{\end{eqnarray}}
\newcommand{\balign}{\begin{align}}
\newcommand{\ealign}{\end{align}}
\newcommand{\ba}{\begin{array}}
\newcommand{\ea}{\end{array}}
\newcommand{\bfig}{\begin{figure}}
\newcommand{\efig}{\end{figure}}

\allowdisplaybreaks

\newcommand{\exval}[1]{\mbox{$\langle \, {#1}\, \rangle$}}

\newcommand{\rmd}{\mathrm{d}}
\newcommand{\rme}{\mathrm{e}}

\newcommand{\R}{{\mathbb R}}

\newcommand{\Z}{{\mathbb Z}}

\newtheorem{theorem}{Theorem}[section]

\newtheorem{definition}[theorem]{Definition}
\newtheorem{proposition}[theorem]{Proposition}
\newtheorem{corollary}[theorem]{Corollary}
\newtheorem{remark}[theorem]{Remark}

\def\qed{\hfill$\Box$\par\medskip\par\relax}

\begin{document}

\title{Mesoscale mode coupling theory for the weakly asymmetric simple exclusion process}
\author{Gunter M. Sch\"utz}
\maketitle

{\small 
Departamento de Matemática,
Instituto Superior Técnico,
Av. Rovisco Pais,
1049-001 Lisboa,
Portugal
\\
\noindent Email: gunter.schuetz@tecnico.ulisboa.pt
%
}

\begin{abstract}
The asymmetric simple exclusion process and its analysis by 
mode coupling theory (MCT) is reviewed. To treat the weakly 
asymmetric case at large space scale $x\varepsilon^{-1}$, 
large time scale $t \varepsilon^{-\chi}$ and weak hopping bias $b 
\varepsilon^{\kappa}$ in the limit $\varepsilon \to 0$ we develop a 
mesoscale MCT that allows for studying the crossover at $\kappa=1/2$ 
and $\chi=2$ from Kardar-Parisi-Zhang (KPZ) to Edwards-Wilkinson 
(EW) universality. The dynamical structure function is shown to satisfy
for all $\kappa$ an integral equation that is independent of 
the microscopic model parameters and has a solution that yields a 
scale-invariant function with the KPZ dynamical exponent $z=3/2$
at scale $\chi=3/2+\kappa$ for $0\leq\kappa<1/2$ and for $\chi=2$
the exact Gaussian EW solution with $z=2$ for $\kappa>1/2$. 
At the crossover point it is a function of both scaling variables which 
converges at macroscopic scale to the conventional MCT approximation 
of KPZ universality for $\kappa<1/2$. This fluctuation pattern 
confirms long-standing conjectures for $\kappa \leq 1/2$
and is in agreement with mathematically rigorous results 
for $\kappa>1/2$ despite the
numerous uncontrolled approximations on which MCT is based.
\end{abstract}

\section{Introduction}
\label{sec:1}

The asymmetric simple exclusion process (ASEP), introduced into the 
mathematical literature in the seminal paper by Spitzer \cite{Spit70}, 
is arguably the conceptually most important stochastic interacting 
particle system \cite{Derr98,Ligg99,Kipn99,Schu01}. Informally this 
continuous-time Markov process can be described as follows. Each site 
$k$ of the integer lattice is occupied by at most one particle (exclusion) 
so that the state of the particle system at any time $\tau$ is 
represented by occupation numbers $\eta(k)\in\{0,1\}$ for $k\in\Z$. 
All particles have an associated alarm clock. When the clock of a 
particle on site $k\in\Z$ rings, which happens independently of all other 
clocks after an exponential random time of mean 1, the particle 
attempts to jump to site $k+1$ with probability $(1+B)/2$ and to site 
$k-1$ with probability $(1-B)/2$ where $B$ is the bias that defines 
the strength of the asymmetry. If the target site $k\pm1$ is occupied 
the jump attempt fails. The process is called symmetric simple 
exclusion process (SSEP) for $B=0$ and totally asymmetric simple 
exclusion process (TASEP) for maximal bias $|B|=1$. The process 
considered on large space scales $k=x\varepsilon^{-1}$, large time 
scales $\tau=t\varepsilon^{-\chi}$ and weak bias $B=b 
\varepsilon^{\kappa}$ is called weakly asymmetric simple exclusion 
process (WASEP) for $\kappa >0$. 

Over the last five decades the ASEP has become a source of deep 
insights into generic behaviour of non-equilibrium interacting particle 
systems in one dimension and is thus sometimes dubbed the Ising 
model of non-equilibrium physics. The Markov generator is closely 
related to the quantum Hamiltonian of the spin-1/2 Heisenberg chain 
\cite{Alex78,Gwa92} which is completely integrable in the Yang-Baxter 
sense \cite{Baxt82,Gaud14}. This link allows for using the algebraic 
machinery involving non-Abelian symmetries, matrix product 
approaches and Bethe ansatz in its various guises to obtain exact 
results particularly on the microscopic properties of the model 
\cite{Derr98,Schu01} and has given prominence to the field of 
integrable probability \cite{Boro16}. 

Complementary probabilistic methods have led 
to a detailed and often mathematically rigorous understanding of the 
emergence of large-scale behaviour \cite{Spoh91,Ligg99,Kipn99,Bert15}. 
The coarse-grained particle density of the SSEP evolves under 
scaling of space and time with $\chi=2$ according to the diffusion 
equation and also the fluctuations around this deterministic limit are 
diffusive with the dynamical exponent $z=2$ of the Edwards-Wilkinson 
(EW) universality class \cite{Edwa82}. The coarse-grained density of 
the ASEP, on the other hand, evolves under Eulerian scaling 
($\chi=1$) according to the Burgers equation which exhibits shocks.
Density fluctuations on the finer scale $\chi=3/2$ belong to the 
celebrated superdiffusive Kardar-Parisi-Zhang (KPZ) universality class 
with dynamical exponent $z=3/2$ \cite{Kard86,Halp15}. Over time, a 
close relationship to directed polymers and growth models has 
revealed deep links between the ASEP, random matrix theory and the 
Kardar-Parisi-Zhang equation \cite{Krie10,Corw12,Spoh17,Mate21}. 

The WASEP with $\kappa=1$ conditioned on carrying an atypically low 
current exhibits a dynamical phase transition from a shock-like 
travelling wave behaviour to the flat density profile characteristic for 
typical currents \cite{Bodi05,Beli13}, while for atypically high current 
in the ASEP one finds a flat profile with algebraically decaying 
stationary density correlations and dynamical exponent $z=1$ 
\cite{Spoh99,Popk11,Jack15a} that have their origin in conformal 
invariance \cite{Kare17} that describes equilibrium critical 
phenomena in two-dimensions \cite{Card10}, thus bringing in view 
yet another seemingly unrelated field of research. 

In the quantum domain the matrix product ansatz, originally developed
for the ASEP \cite{Derr93a}, was generalized to the open Heisenberg 
spin chain \cite{Pros11,Pros15} and has revealed intriguing stationary 
states with density profiles \cite{Pros11,Kare13b} and long-range 
correlations \cite{Pros11,Buca16} that are reminiscent of the open 
SSEP 
\cite{Spoh83,Derr02,Derr07,Bert15}. Developping the Bethe ansatz 
further has led to generalized hydrodynamics \cite{Cast16,Bert16} 
which captures the large-scale temporal evolution of the quantum 
Heisenberg chain and related integrable many-body quantum system on 
Eulerian scale by taking into account the presence of infinitely many 
conservation laws \cite{Doyo20,Essl22}. Interestingly, also in the 
quantum setting fluctuations exhibit scaling behaviour known 
from the classical KPZ universality class \cite{Ljub19,DeNa19,Ye22}.

Remarkably, despite its simplicity the one-dimensional ASEP and its 
specializations SSEP and TASEP are important models also from an 
experimental perspective. The SSEP is a model for interface 
fluctuations in the universality class of the Edwards-Wilkinson equation 
with numerous experimental realizations \cite{Edwa82,Krug97}. It also 
serves as a model for reptation dynamics of entangled polymers 
\cite{Doi85} and has been shown \cite{Schu99} to explain 
experimental data on the tube-like motion of a single polymer-chain 
\cite{Perk94}. A single tagged particle in the SSEP has a subdiffusive 
mean square displacement \cite{Arra83,vanB83} which has been 
observed in single-file diffusion in zeolites \cite{Kukl96} and colloidal 
systems \cite{Wei00}. In the presence of pair annihilation the 
generator of the process can be expressed in terms of free fermions 
\cite{Lush87,Gryn95,Schu95,benA98} and thus allows for a very 
detailed rigorous study of diffusion-limited annihilation characterized in 
one dimension by an anomalous algebraic decay of the particle density 
which had been discovered earlier in exciton dynamics on polymer 
chains \cite{Shan82,Vard82,Priv97}.

The ASEP exhibits shocks both microscopically and on macroscopic 
Eulerian scale \cite{Reza91,Ferr91,Derr93b,Kreb03} with diffusive 
fluctuations of the shock position \cite{Ferr94a,Beli02} that clarify 
\cite{Schu97a} an apparent slowing down of ribosomes while moving 
on a mRNA template during protein synthesis, a fundamental biological 
process for which the ASEP was originally conceived as a qualitative 
model \cite{MacD68}. Perhaps the most spectacular experimental 
realization of ASEP behaviour is a series of beautiful experiments on 
turbulent liquid crystal systems by Takeuchi and Sano where interface 
fluctuations follow with remarkable accuracy the universal behaviour 
expected from the KPZ universality class \cite{Take10,Take11}. 
Predictions from generalized hydrodynamics for integrable quantum 
systems have been tested in -- by now -- numerous experiments with 
cold atoms \cite{Vidm15,Lang15,Kerr23}. A novel feature of the 
Heisenberg quantum chain that seems to have no analogue in the 
classical ASEP, but has a mathematical structure in common with 
shock measures, is the appearance of spin helix states, originally 
explored theoretically in \cite{Popk17,Popk21} and recently discovered 
in experiments on anisotropic Heisenberg chains \cite{Jeps22}.

This brief survey of exact, rigorous and experimental literature 
published from 1968 until today can only hint at the reasons why the 
ASEP continues to fascinate. Indeed, it might come as a surprise that 
there are still open questions related to the ASEP and 
its relationship to KPZ universality. From an experimental viewpoint
that most pressing one is that despite the mathematically expected 
robustness of the KPZ universality class the corresponding asymptotic
behaviour of experimentally accessible quantities is rarely seen in 
classical systems. The examples given above (and some more) are
notable exceptions. This raises the question of corrections to scaling.
Recent numerical evidence suggests that diffusive corrections are 
important \cite{Schm21}, but otherwise this question is essentially 
completely open. 

A second major open problem closely related to the ASEP is the 
rigorous derivation of dynamical universality classes in processes with 
more than one conservation law. This issue has been addressed 
systematically for two conservation laws with the mode coupling 
approximation \cite{Spoh14,Popk15a,Spoh15} which revealed the 
existence of new universality classes beyond the ubiquitous diffusion 
and KPZ. Subsequent work on an arbitrary number of conservation 
laws showed that all these universality classes are the members of a 
discrete infinite family \cite{Popk15b}. This series of dynamical
universality classes was dubbed the Fibonacci family because the 
dynamical exponents are all Kepler ratios of neighbouring Fibonacci 
numbers $F_i$. The series begins with Edwards-Wilkinson (EW)
($z=2=F_3/F_2$), is then followed by (a) KPZ, (b) a possibly new 
universality class related to KPZ, and (c) a third universality class 
named after L\'evy, all with $z=3/2=F_4/F_3$, then the heat 
mode exponent $z=5/3=F_5/F_4$ \cite{vanB12}, and then the 
remaining series $z_i=F_{i+1}/F_{i}$ up to the limit value 
$z_{\infty}=\varphi:=(1+\sqrt{5})/2$ which is the famous golden 
mean. 

All of these mode coupling studies consider the large-scale limit 
$x\to\infty$, $t\to\infty$ for strong asymmetry $\varepsilon=0$. 
An open problem is the question is the extent to which these
predictions are correct. Numerical Monte-Carlo studies confirm the 
predicted universal scaling functions \cite{Popk15b}, but not 
generally the non-universal scale factors. It therefore remarkable
that in one particular case rigorous work on a chain of oscillators
\cite{Bern12,Bern16} fully confirms the mode coupling prediction.
In a related series of rigorous work on two conservation laws
with weak asymmetry the KPZ and L\'evy universality classes with
$z=3/2$ as well as the  EW class with $z=2$ were found 
\cite{Bern14,Bern18,Ahme22}, but the present state of art of
the MCT analysis for strong asymmetry is not applicable and can 
therefore not be tested against these results.

This leads to a third and more subtle open issue which is the strong 
universality conjecture and concerns the exact relationship between 
the KPZ fixed point (characterizing the KPZ universality class) and the 
WASEP which exhibits a transition from KPZ to EW behaviour at a 
critical scale of a weak bias \cite{Gonc12} made more precise below. 
It is the purpose of this paper to extend the mode coupling 
approximation to the WASEP to address this issue from the mode 
coupling perspective. The main new contributions of the present work 
are to point out the existence on a mesoscopic scale of generalized 
mode coupling solutions applicable to the weakly asymmetric case 
(and then by inference to any system with one conservation law that is 
governed by nonlinear fluctuating hydrodynamics) and to elucidate 
under which conditions on the microscopic parameters these solutions 
are scale invariant.

To put these results into perspective we first continue the review of 
the ASEP, but with a focus on weak asymmetry (Sec. 2). In Sec. 3
the conventional mode coupling approximation is introduced and the 
underlying assumptions are made explicit. The mesoscopic MCT is 
developed in Sec. 4 and the asymptotic solution for the weakly 
asymmetric case is proved and commented upon. A brief summary 
and some final remarks are presented in Sec. 5.

\section{The (W)ASEP and its relation to KPZ and EW universality}

For any bias $B$ the unique invariant measure of the ASEP that is
also translation invariant is the Bernoulli product measure with 
density $\rho$ and static compressibility $K(\rho) = \rho(1-\rho)$ 
\cite{Ligg99}. For $B\neq 0$ this invariant measure is not reversible 
so that the ASEP defines a particle system that is perpetually out of 
thermal equilibrium with stationary particle current $j(\rho) = B 
\rho(1-\rho)$. The derivatives
\begin{equation}
c(\rho) := \frac{\rmd}{\rmd \rho} j(\rho) = B (1-2\rho), \quad
\check{G}(\rho) := \frac{\rmd^2}{\rmd \rho^2} j(\rho) = -2 B
\label{ccheckGdef}
\end{equation}
are the collective velocity $c(\rho)$ which is the center-of-mass 
velocity of the density fluctuations \cite{Schu01} and the coupling 
strength $\check{G}(\rho)$ that is a measure of the strength of the 
nonlinearity of the current as a function of the particle density 
$\rho$. Below we drop
the dependence of functions on $\rho$ to avoid excessive notation.

Consider the centered occupation numbers $\bar{\eta}(k) := 
\eta(k) - \rho$ and expectations w.r.t. the Bernoulli product measure 
with density $\rho$. The main quantity of interest in the present 
context is the large-scale behaviour of the stationary covariance
\begin{equation}
C(k,\tau) := \exval{\bar{\eta}(k,\tau)\bar{\eta}(0,0)}, \quad k\in\Z
\label{Covktaudef}
\end{equation}
and its Galilei-shifted Fourier transform
\begin{equation}
S(q,\tau) = \rme^{i c \sin{(q)}\tau} \frac{1}{\sqrt{2\pi}} 
\sum_{k\in\Z} \rme^{-iq k} C(k,\tau), \quad q \in (-\pi,\pi] 
\label{DSFqtaudef}
\end{equation}
which is called the dynamical structure function (DSF). 
Since $\eta(k)\in\{0,1\}$ and $\rho = \exval{\eta(k)} \in [0,1]$ the 
covariance $C(k,\tau)$ is trivially uniformly bounded from above 
and below by $\pm1$. Due to the Lieb-Robinson bound \cite{Lieb72} 
the covariance decays quickly at distances far from the center-of-mass 
position $k^\ast = [ct]$. Thus $S(q,\tau)$ is well-defined and has 
``nice'' properties for all $q\in\R$ and $\tau\in\R_0^+$. 

The collective velocity $c$ can be removed by a Galilei 
transformation, i.e., by studying the properties of the process in a 
shifted coordinate frame moving with this velocity.
Reflection symmetry then yields
\begin{equation}
S(q,\tau) = S(-q,\tau), \quad c = 0.
\label{refsym}
\end{equation}
For $\tau=0$
the Bernoulli product measure also yields
\begin{equation}
C(k,0) = K \delta_{k,0}, \quad 
S(q,0) = S(0,\tau) = \frac{K}{\sqrt{2\pi}}.
\label{Covtau0}
\end{equation}
Hence it is sufficient to restrict the discussion to the domain
$q,\tau\in\R^+$ of strictly positive momentum and time.

We study the large scale behaviour of the dynamical structure function 
by taking the limit $\varepsilon\to 0$ for $q = \varepsilon p$ and 
$\tau = \varepsilon^{-\chi} t$ with fixed strictly positive $p$ and $t$. 
The parameter $\chi>0$, which we call observation scale, quantifies 
how much one speeds up the microscopic time. The hopping bias $B$ 
of the ASEP comes in by taking $B=b \varepsilon^{\kappa}$.
We call the bias strong for $\kappa = 0$ and weak for $\kappa > 0$, 
with the qualification ``moderately weak'' for the range 
$0 < \kappa < 1/2$. Universality means that at some critical
observation scale $\chi^\ast$ the limiting dynamical structure 
function 
\begin{equation}
S^\ast(p,t) := \lim_{\varepsilon \to 0} 
S(\varepsilon p,\varepsilon^{-\chi} t)
\end{equation}
is such that the ratio $S^\ast(p,t)/S^\ast(0,0)$ is a non-trivial function 
only of a scaling variable $u=l p t^{1/z}$ with dynamical exponent 
$z$ and non-universal scale parameter $l$. The microscopic details
of the model, i.e., of the density $\rho$ and the bias parameters 
$b$ and $\kappa$ enter only $l$, in a
parameter submanifold that is not only a point in the three-dimensional
space of the microscopic parameters $\rho,b,\kappa$.

For strong asymmetry the asymptotic dynamical structure function of 
the KPZ universality class (to which the ASEP fluctuations belong, see 
above) was first studied using mode coupling approximation 
\cite{vanB85} and shown at observation scale $\chi=3/2$ to be a 
scaling function with dynamical exponent $z=3/2$. A numerical 
evaluation of the mode coupling solution \cite{Frey96,Cola01} 
shows a clear difference to the Gaussian that one obtains for the SSEP 
at observation scale $\chi=2$ and which characterizes the 
diffusive EW universality class with dynamical exponent $z=2$. The 
exact asymptotic dynamical structure function of the ASEP remained 
elusive for a long time, but was eventually derived 
in a remarkable analysis of the polynuclear growth model 
\cite{Prae04,FerrPL06}. Its universal scaling form 
is given by the Fourier transform of the
Pr\"ahofer-Spohn scaling function $f_{PS}(\cdot)$.

For weak asymmetry it has been known for a long time that for 
$\kappa =1$ and on diffusive observation scale $\chi=2$ the 
limiting density fluctuations of the WASEP converge to a generalized 
Ornstein-Uhlenbeck process \cite{DeMa89,Ditt91} with a drift 
term proportional to $v=b (1-2\rho)$ which can be removed by 
a Galilei transformation $x \to x+vt$. The corresponding universal
dynamical structure function is then the Gaussian with diffusion 
constant $D=1/2$. This result was later extended to the range 
$1/2<\kappa \leq 1$ \cite{Gonc12} and shows that independently
of $\kappa$ (in that range) the density fluctuations belong to the 
Edwards-Wilkinson universality class 
\cite{Edwa82} with dynamical exponent $z=2$ like the SSEP. 

For $\kappa=1/2$ one obtains after a Hopf-Cole transformation 
\cite{Gart88} an exponential process for which the limiting 
fluctuations are given by the stochastic heat equation \cite{Bert97}. 
Looking directly at the fluctuation field (i.e., without Hopf-Cole 
transformation) it was proved \cite{Gonc12} that the limiting  
fluctuation field $\phi(x,t)$ satisfies at observation
scale $\chi=2$  (weakly) the 
stochastic Burgers equation (SBE) which we write loosely in the 
intuitive form 
\begin{equation}
\partial_t \phi = D \partial_x^2 \phi + 2 \check{G}^2 (\partial_x \phi)^2
+  \partial_x \xi
\label{SBE}
\end{equation}
with Gaussian space-time white noise $\xi$.

The SBE, related to the KPZ equation \cite{Kard86} for the height
function $h(x,t)$ via $\phi(x,t)=\partial_x h(x,t)$, is the generic 
evolution equation for systems with one conservation law that are 
governed by nonlinear fluctuating hydrodynamics \cite{Spoh91}. The 
fundamental underlying assumption is the validity of a 
Boltzmann-Gibbs principle which assumes that the relaxation of the 
dynamical structure function is governed by the long-wave length 
Fourier modes of the conserved field. It should be noted that the 
fluctuation field $\phi(x,t)$ appearing in the SBE is {\it not} scale 
invariant under the transformation $x\to\lambda x$, 
$t\to\lambda^{3/2}t$ as one might naively expect from KPZ 
universality and which is true for the Markov field that defines the 
KPZ fixed point and is characterized by its transition probabilities 
\cite{Mate21}. However, this scaling is expected to send $\phi(x,t)$ 
to the true universal fixed point. While universality is by now 
well-established \cite{Quas23} it is not known how this is reflected 
in the properties of the dynamical structure function for $\kappa=1/2$.

For $b\to \infty$ and likewise in the range $0 < \kappa <1/2$ 
for fixed $b$ one expects the same universal KPZ fluctuations as for 
the strongly asymmetric case $\kappa=0$ \cite{Amir11}. 
The exact spectral gap of the (W)ASEP computed from Bethe ansatz
\cite{Kim95} suggests that these KPZ fluctuations appear on observation scale $\chi^\ast = 3/2+\kappa$, see also an argument for
a system with two conservation laws that has been deduced 
\cite{Cannarx23} from the universal KPZ scaling discussed in 
\cite{Quas23}. 
A rigorous proof of these conjectures is lacking, however.
\footnote{An old result from coupling theory for the one-dimensional ASEP \cite{vanB85} also suggests $\chi^\ast = 3/2+\kappa$
but does not directly predict the range $0 < \kappa <1/2$
nor the behaviour at $\kappa = 1/2$
unlike the mesoscale mode coupling theory developed below.}

\section{Mode coupling theory for the ASEP}

To discuss the time evolution of the (W)ASEP it is convenient to
split the generator $\mathcal{L}^{ASEP}$ into a part $\mathcal{L}^{ASEP}_0$ that acts 
linearly on the occupation numbers $\eta(k)$ and a nonlinear
part $\mathcal{L}^{ASEP}_1$ that would be absent in a non-interacting particle system, i.e.,  
\begin{equation}
\mathcal{L}^{ASEP} = \mathcal{L}^{ASEP}_0 +
\mathcal{L}^{ASEP}_1.
\label{geneta}
\end{equation}
With the symmetric discrete gradient 
$\nabla f(k) := [f(k+1)-f(k-1)]/2$, the asymmetric discrete gradients
$\nabla^+ f(k) := f(k+1)-f(k)$, $\nabla^- f(k) := f(k)-f(k-1)$, the
discrete Laplacian $\Delta f(k) := f(k+1)+f(k-1) - 2 f(k)$ and
the fixed diffusion coefficient $D=1/2$ the action of the generator
on the centered occupation numbers is thus given by \cite{Ligg99}
\begin{equation}
\mathcal{L}^{ASEP}_0 \bar{\eta}(k) = D \Delta \bar{\eta}(k) 
- c \nabla \bar{\eta}(k) 
\label{genASEPlin}
\end{equation}
and
\begin{equation}
\mathcal{L}^{ASEP}_1 \bar{\eta}(k) = \check{G} \bar{\eta}(k) 
\nabla \bar{\eta}(k) = \check{G} \nabla^- \bar{\eta}(k) \bar{\eta}(k+1)
\label{genASEPnl}
\end{equation}
with the collective velocity $c$ and coupling strength $\check{G}$
defined in (\ref{ccheckGdef}).

The first approximation is to replace the binary random number 
$\bar{\eta}(k)$ by a real-valued random variable $\phi(k)$ and 
interpret the action (\ref{genASEPlin}) and (\ref{genASEPnl}) on 
$\bar{\eta}(k)$ as representing the action of a generator
$\mathcal{L} = \mathcal{L}_0 + \mathcal{L}_1$ of some discrete 
nonlinear stochastic evolution equation for $\phi(k)$ with an explicit 
and simple invariant measure. This is achieved with the choice
\begin{eqnarray}
\mathcal{L}_0 & = & \sum_{j\in\Z} \left[- v \nabla \phi(j) + 
D \Delta \left(\phi(j) - \partial_{\phi(j)}\right)\right] 
\partial_{\phi(j)} \\
\mathcal{L}_1 & = & \check{G} \sum_{j\in\Z} \left[\left(\nabla \phi(j) \right) \left(\phi(j) + 
\frac{1}{3} \Delta \phi(j)\right)\right] 
\partial_{\phi(j)} 
\end{eqnarray}
which has the desired property
\begin{equation}
\mathcal{L}_0 \phi(k) = - c \nabla \phi(k) + 
D \Delta \phi(k)
\label{gendSBElin}
\end{equation}
for the linear part mimicking (\ref{genASEPlin}) of the linear
evolution in the ASEP.
 The nonlinear part (\ref{genASEPnl}) is approximated as
\begin{equation}
\mathcal{L}_1 \phi(k) = \check{G} \frac{\phi(k-1) + \phi(k) + 
\phi(k+1)}{3} \nabla \phi(k) = \check{G} \nabla^- F(k)
\label{gendSBEnl}
\end{equation}
with the average $[\phi(k-1) + \phi(k) + \phi(k+1)]/3 = 
(1+ \Delta/3) \phi(k) $ rather
than just $\phi(k)$ in (\ref{genASEPnl})
and with
\begin{equation}
F(k) = \frac{\phi^2(k) + \phi(k) \phi(k+1)
+ \phi^2(k+1)}{3} =
\phi(k) \phi(k+1) + \frac{1}{3} \left(\nabla^- \phi(k)\right)^2.
\label{dSBEnl}
\end{equation}
This discretization ensures that the invariant measure of 
$\mathcal{L}$ generator is a product of mean-zero normal 
distributions. As noted in \cite{Spoh14} this is the generator 
associated with a discretized SBE (\ref{SBE}) with Gaussian 
white noise.

With this SBE approximation the space-time covariance 
(\ref{Covktaudef}) satisfies the evolution equation
\begin{eqnarray}
\frac{\rmd}{\rmd \tau} C(k,\tau) 
& = & \exval{\phi(0)
\left[
\rme^{\mathcal{L}_{0} \tau} + 
\int_{0}^{\tau} \rmd \sigma \, \rme^{\mathcal{L}_{0} (\tau - \sigma)} 
\mathcal{L}_{1} \rme^{\mathcal{L} \sigma}\right] (\mathcal{L}_{0}
+\mathcal{L}_{1}) \phi(k)} \\
& = & \exval{\phi(0)
\rme^{\mathcal{L} \tau} \mathcal{L}_{0} \phi(k)} 
+ \exval{\phi(0)\rme^{\mathcal{L}_{0} \tau} \mathcal{L}_{1} \phi(k)}\nonumber \\
& & + \int_{0}^{\tau} \rmd \sigma \, \exval{\phi(0)
\rme^{\mathcal{L}_{0} (\tau - \sigma)} \mathcal{L}_{1} 
\rme^{\mathcal{L}\sigma}\mathcal{L}_{1} \phi(k)}
\end{eqnarray}
where the expectation is taken w.r.t. to the Gaussian product measure.
However, since the action of $\mathcal{L}_{1} \phi(k)$ is bilinear
in the fields $\phi(\cdot)$ and all odd moments vanish in the expectation, the second
term in r.h.s. of the second line vanishes. Next, notice that
the adjoint of $\mathcal{L}_{1}$ w.r.t. to the flat volume measure 
has the simple property $\mathcal{L}^T_{1} = - \mathcal{L}_{1}$
and that $\mathcal{L}_{0}$ is linear. By using translation
invariance and the second equality in (\ref{gendSBEnl})
the evolution equation reduces to 
\begin{equation}
\frac{\rmd}{\rmd \tau} C(k,\tau) 
= \left(- c \nabla + 
D \Delta\right) C(k,\tau)
+  \check{G}^2 \Delta \sum_{l\in\Z} \int_{0}^{\tau} \rmd \sigma \,  
C_0(l-k,\tau-\sigma) \Omega(l,0;\sigma)
\label{apprx1}
\end{equation}
where 
\begin{equation}
C_0(k-l,\tau-\sigma) = \exval{\phi(0)
\rme^{\mathcal{L}_0 (\tau-\sigma)} \phi(l-k)}
\label{dSBElinprop}
\end{equation}
is the linearized covariance given by integrating 
(\ref{gendSBElin}) and  
\begin{equation}
\Omega(l,m;\sigma) := \exval{F(l)
\rme^{\mathcal{L}\sigma} F(m)}.
\label{Omegadef}
\end{equation}
with $F(\cdot)$ defined in (\ref{dSBEnl}). This is an exact evolution
equation for the covariance of the discretized SBE, but not directly
tractable due the appearance of the four-point function (\ref{Omegadef}).
To proceed we follow
\cite{vanB85,Spoh14} and introduce the
various approximations that lead to the mode coupling equation
for the space-time covariance (\ref{Covktaudef}).

\paragraph{\underline{MCT Approximation 1 (Gaussian random vector):}}

The idea to handle (\ref{Omegadef}) is to treat the random numbers
$\phi(k,\tau)$ as components of a mean-zero Gaussian
random vectors which allows for expressing expectations of a 
product of an even number of components by the covariances
\begin{equation}
\exval{\phi(k_1,\tau_1)\dots \phi(k_m,\tau_m)} = \sum_{p \in P^{(2)}_n}
\prod_{(i,j)\in p} \exval{\phi(k_i,\tau_i)\phi(k_j,\tau_j)}
\label{Isserlis} 
\end{equation}
where $P^{(2)}_n$ is the set of all the distinct pairings $(i,j)$ 
of the set $\{1,\ldots,m\}$. Expectations of an odd number of
components are zero. 

It is easy to see that in the subsequent computations the gradient
part in (\ref{dSBEnl}) generates only subleading contributions
to the asymptotic behaviour and can therefore be neglected.
Using translation invariance and stationarity and neglecting further subleading gradient terms then yields
\begin{equation}
\Omega(l,0;\sigma) 
\approx 2 \exval{\phi(l,\sigma) \phi(0,0)}^2 .
\end{equation}
and therefore
\begin{equation}
\frac{\rmd}{\rmd \tau} C(k,\tau) 
= \left(- c \nabla + 
D \Delta\right) C(k,\tau)
+  2 \check{G}^2 \Delta \sum_{l\in\Z} \int_{0}^{\tau} \rmd \sigma \, C_0(l-k,\tau-\sigma) 
C^2(l,\sigma)
\label{apprx2}
\end{equation}
involving both the full covariance and the linear
evolution of the  covariance.

\paragraph{\underline{MCT Approximation 2 (Closure):}}

To close the equation (\ref{apprx2}) the linear evolution 
(\ref{dSBElinprop}) is replaced {\it ad hoc} by the full covariance.
This yields the mode coupling equation
\begin{equation}
\frac{\rmd}{\rmd \tau} C(k,\tau) 
= \left(- c \nabla + 
D \Delta\right) C(k,\tau)
+  2 \check{G}^2 \Delta \sum_{l\in\Z} \int_{0}^{\tau} \rmd \sigma \,  C(k-l,\tau-\sigma) 
C^2(l,\sigma) 
\label{apprx3}
\end{equation}
which is a discrete version of the mode coupling equation
originally obtained in \cite{vanB85}.

We recall that the discrete gradient term that remains in 
(\ref{apprx3}) can be removed by a Galilei transformation. 
For the dynamical structure function in this moving coordinate
frame Fourier transformation
yields with the eigenvalue $\mathsf{e}(q) := 2 [1 - \cos{(q)}]$
of the discrete Laplacian
\begin{eqnarray}
\frac{\rmd}{\rmd \tau} S(q,\tau) 
& = & 
- D \mathsf{e}(q) S(q,\tau) 
\nonumber \\ & & 
- \mathsf{e}(q) 2 \check{G}^2
\int_{0}^{\tau} \rmd \sigma \,  S(q,\tau-\sigma)  
\int_{-\pi}^{\pi} \rmd q' \, 
S(q',\sigma) S(q-q',\sigma).
\label{dMCT}
\end{eqnarray}
Since according to (\ref{Covtau0}) $S(q,0)=K/\sqrt{2\pi}$ 
a Laplace transformation
\begin{equation}
\tilde{S}(q,s) = \int_{0}^\infty \rmd \tau \, \rme^{-s \tau}
S(q,\tau)
\end{equation}
renders the mode coupling equation (\ref{dMCT}) in the equivalent
form
\begin{equation}
\frac{K}{\sqrt{2\pi}}
= \left[s
+ D \mathsf{e}(q) + 2 \check{G}^2 \mathsf{e}(q) \tilde{M}(q,s)
\right] \tilde{S}(q,s) 
\label{dMCTLFT}
\end{equation}
with the memory integral
\begin{equation}
\tilde{M}(q,s) = 
\int_{0}^{\infty} \rmd \tau \, \rme^{-s \tau} 
\int_{-\pi}^{\pi} \rmd q' \, 
S(q',\tau) S(q-q',\tau)
\label{memLFT}
\end{equation}
arising from the nonlinearity of the process. The equation is reflection 
symmetric and trivially satisfied for $q=0$ and hence consistent with 
the exact properties (\ref{refsym}) and (\ref{Covtau0}). This 
Fourier-Laplace representation of the mode coupling equation is the 
starting point for further developing the mode coupling analysis.

\section{Mesoscale mode coupling theory}

In the original paper \cite{vanB85} treating the ASEP ($\kappa=0$) 
as well as in the later work
for strongly asymmetric systems with more than one conservation 
law \cite{Spoh14,Popk15a,Popk15b} the discrete mode coupling 
approximation (\ref{apprx3}) and consequently the FT (\ref{dMCT})
are further approximated by taking $k\in\R$ and replacing the
discrete gradient and discrete Laplacian by the usual 
partial derivative $\partial_k$ and $\partial_k^2$ resp.
For the FT this means that $q\in\R$ and the integration over $q'$
extends over the whole of $\R$. The second step in these 
works is then to look for asymptotic scaling solutions for 
and $q\to 0$, $t\propto q^{-z}\to\infty$ with dynamical exponent $z$, 
corresponding to macroscopic space and time scales. Here we take 
an approach that is different in several respects.

\begin{itemize}
\item[1)] Non-scaling solutions: Rather than focussing on strong asymmetry and looking for asymptotic scaling behaviour small 
$q$ and some fixed $z$ we do not make any scaling assumption.

\item[2)] Weak asymmetry: The asymmetry exponent $\kappa$
may take any positive value, including the stronlgy asymmetric case $\kappa=0$.

\item[3)] Mesoscopic and macroscopic scale:  We choose $q 
= \varepsilon p$, $\tau = \varepsilon^{-\chi} t$ and take the limit 
$\varepsilon \to 0$ keeping $p\in\R^+$, $t\in\R^+$ fixed and with
$\chi\in\R^+$ taken as a variable.

\item[4)] Change of variables: Instead of the ``natural''
variables $p,t$ we analyse the mode coupling equation (\ref{dMCT}) 
in terms of judiciously chosen variables that depend on the microscopic 
parameters $D,b,\kappa$ of the WASEP.

\item[5)] No continuum approximation: The starting point for 
investigating solutions of the mode coupling equation is the 
discrete equation (\ref{dMCTLFT}).
\end{itemize}

The observation exponent $\chi$ in item 3
defines the space-time scale on which the (W)ASEP is investigated and 
which we call mesoscopic, as opposed to a macroscopic scale defined 
by a second limit, viz., $p\to \lambda p$ with $\lambda\to 0$ and time 
$t \to \lambda^{-a} t$ accelerated at some rate defined by the 
macroscale exponent $a$. The main points of the mesoscale approach, expressed in items 1 - 3, 
are that we neither assume scale invariance from the outset and nor 
restrain ourselves to the macroscopic scale and strong asymmetry. 
Instead we determine the parameter ranges 
for which scale invariant solutions appear (or not) on the mesoscopic 
scale for arbitrary $\kappa\geq0$. For scale-invariant dynamical structure functions the solutions 
are the same on the mesoscopic and the macroscopic scale, but they 
are different in the absence of scale invariance. The mesoscopic 
description allows for considering this case and for studying 
crossover between different dynamical universality classes. 
The technical item 4 facilitates studying the mesocopic limit.
Item 5 is mathematically immaterial, but helps keeping track of 
the physical interpretation of the various space-time scales 
(microscopic, mesoscopic and macroscopic) under consideration.
This distinction is made precise as follows.

\begin{definition}
\label{Def:mesomacroscale}
The mesoscopic DSF is the limit 
\begin{equation}
S^\ast(p,t) := \lim_{\varepsilon \to 0} 
S(\varepsilon p, \varepsilon^{-\chi} t), \quad
\tilde{S}^\ast(p,\omega) := \lim_{\varepsilon \to 0} \varepsilon^{\chi} 
\tilde{S}(\varepsilon p, \varepsilon^{\chi} \omega)
\label{Smesodef}
\end{equation}
of the microscopic DSF $S(\cdot, \cdot)$, $\tilde{S}(\cdot, \cdot)$ 
in the Fourier and Fourier-Laplace domain respectively.
The macroscopic DSF is the limit 
\begin{equation}
S^{\bullet}(p,t) := \lim_{\varepsilon \to 0} 
S^\ast(\varepsilon p, \varepsilon^{-a} t), \quad
\tilde{S}^{\bullet}(p,\omega) := \lim_{\varepsilon \to 0} \varepsilon^{a} 
\tilde{S}^\ast(\varepsilon p, \varepsilon^{a} \omega)
\label{Smacrodef}
\end{equation}
of the mesoscopic DSF.
\end{definition}
Whether or not the DSF is scale invariant or even trivial is to be determined as a function of the obervation scale parameters $\chi$ 
and $a$ rather than made property of an ansatz as in  \cite{Spoh14,Popk15a,Popk15b}.

\subsection{Change of variables}

We study the dynamical structure function 
\begin{equation}
S_{\varepsilon}(p,t) := S(\varepsilon p, \varepsilon^{-\chi} t)
\end{equation}
as a function of the variables
\begin{equation}
\zeta_1(p,t) := \varepsilon^{(\chi_1-\chi)/z_1} |p| (l_1 t)^{1/z_1}, \quad 
\zeta_2(p,t) := \varepsilon^{(\chi_2-\chi)//z_2} |p| (l_2 t)^{1/z_2}
\label{zetadef}
\end{equation}
with parameters $z_i,l_i,\chi_i$ that depend on the microscopic
parameters $b,D,\kappa$ in a way determined below.
This mapping with the inverse transformation
\begin{equation}
p(\zeta_1,\zeta_2)
= 
\left(\frac{\varepsilon^{\chi_1} l_1 \zeta_2^{z_2}}
{\varepsilon^{\chi_2} l_2 \zeta_1^{z_1}}\right)^{\frac{1}{z_2-z_1}} 
, \quad t(\zeta_1,\zeta_2) = \varepsilon^{\chi}
\left(
\frac{\zeta_1}{\varepsilon^{\chi_1}l_1}
\right)^{\frac{z_2}{z_2-z_1}}
\left(\frac{\varepsilon^{\chi_2}l_2}{\zeta_2} \right)^{\frac{z_1}{z_2-z_1}} .
\label{zetainv}
\end{equation}
is not bijective for $z_2=z_1$ which is therefore excluded
from consideration. Without loss of generality we assume 
$z_2>z_1$.

To express the microscopic dynamical structure function $S_{\varepsilon}(p,t)$
in terms of the variables $\zeta_1,\zeta_2$ 
we introduce the normalized 
microscopic dynamical structure function
$s_{\varepsilon}(\cdot,\cdot)$ by the relation
\begin{equation}
s_{\varepsilon}(\zeta_1, \zeta_2) := \frac{\sqrt{2\pi}}{K} S_{\varepsilon}(p(\zeta_1, \zeta_2),t(\zeta_1, \zeta_2))
\label{hsdef}
\end{equation}
with the constant $K/\sqrt{2\pi}$ derived in (\ref{Covtau0}).
In terms of the
scaling variables
\begin{equation}
u_i \equiv u_i(p,t) := (l_i |p|^{z_i} t)^{1/z_i} 
\label{scalvarFT}
\end{equation}
which are invariant under the scaling transformation $p\to \lambda p$,
$t\to \lambda^{z_i} t$ one has
$\zeta_i(p, t) = 
  \varepsilon^{(\chi_i-\chi)/z_i} u_i(p,t) $ so that
the microscopic DSF $S_{\varepsilon}(p,t)$ is expressed
in terms of the function $s_{\varepsilon}(\cdot,\cdot)$ by
\begin{equation}
S_{\varepsilon}(p,t) =
\frac{K}{\sqrt{2\pi}} s_{\varepsilon}
(\varepsilon^{(\chi_1-\chi)/z_1} u_1(p,t),
\varepsilon^{(\chi_2-\chi)/z_2} u_2(p,t)) .
\label{DSFepspt}
\end{equation}
We note the limits
\begin{equation}
s_{\varepsilon}(0,0) = \lim_{u\to 0} s_{\varepsilon}(ux,uy) = 1, \quad
\lim_{u\to\infty} s_{\varepsilon}(u,y) = 
\lim_{u\to\infty} s_{\varepsilon}(x, u) = 0 
\end{equation}
for fixed $x,y \in \R^+_0$ which arise from (\ref{Covtau0}) and
from the Lieb-Robinson bound respectively. 

To make the corresponding change of variables in the Laplace transform of 
the mode coupling equation we introduce the scaling variables
\begin{equation}
\tilde{u}_i \equiv \tilde{u}_i(p,\omega) := \frac{l_i |p|^{z_i}}{\omega} 
\label{scalvarFLT}
\end{equation}
with the inverse relations
\begin{equation}
p \equiv p(\tilde{u}_1,\tilde{u}_2)
=   \left(\frac{l_1 \tilde{u}_2}{l_2 \tilde{u}_1}\right)^{\frac{1}{z_2-z_1}} , \quad
\omega \equiv \omega(\tilde{u}_1,\tilde{u}_2) = \left(
\frac{l_1}{\tilde{u}_1}
\right)^{\frac{z_2}{z_2-z_1}}
\left(\frac{\tilde{u}_2}{l_2} \right)^{\frac{z_1}{z_2-z_1}} 
\label{scalvarFLTinv}
\end{equation}
and the auxiliary functions
\begin{eqnarray}
\tilde{s}_{\varepsilon}(x,y) & := &
\int_{0}^\infty \rmd \tau \, \rme^{- \tau}
s_{\varepsilon}((\varepsilon^{\chi_1-\chi} x \tau)^{1/z_1},  (\varepsilon^{\chi_2-\chi} y \tau)^{1/z_2}) 
\label{tspmeps} \\
\tilde{m}_{\varepsilon}(x,y) 
& := & \int_{0}^{\infty} \rmd \tau \, \rme^{-\tau} 
\int_{-\varepsilon^{-1}B(x,y)}^{\varepsilon^{-1}B(x,y)} 
\frac{\rmd q}{2\pi} \, 
s_{\varepsilon}(q (\varepsilon^{\chi_1-\chi} x \tau)^{1/z_1},
q (\varepsilon^{\chi_2-\chi} y \tau)^{1/z_2})
\nonumber \\ & & 
\times s_{\varepsilon}((1-q) (\varepsilon^{\chi_1-\chi} x \tau)^{1/z_1},
(1-q) (\varepsilon^{\chi_2-\chi} y \tau)^{1/z_2})
\label{tmpmeps}
\end{eqnarray}
with the integration range in the second line defined by
\begin{equation} 
B(x,y) := \pi \left(\frac{l_1 y}{l_2 x}\right)^{\frac{1}{z_1-z_2}} .
\end{equation} 

Laplace transformation then yields
\begin{eqnarray}
\tilde{S}_{\varepsilon}(p,\omega) & := &
\tilde{S}(\varepsilon p, \varepsilon^{\chi} \omega) \, = \, \frac{K}{\sqrt{2\pi}} \varepsilon^{-\chi} \frac{1}{\omega} 
\tilde{s}_{\varepsilon}(\tilde{u}_1(p,\omega),\tilde{u}_2(p,\omega))
\label{tDSFeps} \\
\tilde{M}_{\varepsilon}(p,\omega) & := &
\tilde{M}(\varepsilon p, \varepsilon^{\chi} \omega) \, = \, K^2 \varepsilon^{1-\chi} \frac{|p|}{\omega}
\tilde{m}_{\varepsilon}(\tilde{u}_1(p,\omega),\tilde{u}_2(p,\omega)).
\label{tmemeps}
\end{eqnarray}
in terms of the scaling variables (\ref{scalvarFLT}).
Therefore with
\begin{equation}
g := \check{G} K \varepsilon^{-\kappa}
= -2 b \rho(1-\rho)
\end{equation}
the discrete mode coupling
equation (\ref{dMCTLFT}) can be written for any fixed $\varepsilon$ in terms of the 
variables $\tilde{u}_1, \tilde{u}_2$ and the function $s_\varepsilon(\cdot,\cdot)$ as
\begin{equation}
\tilde{s}_{\varepsilon}(\tilde{u}_1,\tilde{u}_2)
= \left[1
+ \frac{\mathsf{e}(\varepsilon |p|)}{\varepsilon^2 |p|^2}
\left(\varepsilon^{2-\chi} \frac{|p|^2}{\omega} D  
+ 2 g^2  \varepsilon^{3+2\kappa-2\chi} \frac{|p|^3}{\omega^2}
\tilde{m}_{\varepsilon}(\tilde{u}_1,\tilde{u}_2)
\right)\right]^{-1}
\label{dMCTLFT3}
\end{equation}
where, using the inverse transformation (\ref{scalvarFLTinv}),
\begin{equation}
\frac{|p|^z}{\omega}
= 
\left(\frac{\tilde{u}_1}{l_1}
\right)^{\frac{z_2-z}{z_2-z_1}}
\left(\frac{\tilde{u}_2}{l_2} \right)^{\frac{z-z_1}{z_2-z_1}}.
\label{pzo}
\end{equation}
The mesoscopic description of the WASEP outlined in the next section
is obtained by taking the limit $\varepsilon\to 0$.

\subsection{Mesoscale solution of the mode coupling equation}

To study solutions of the mode coupling equation (\ref{dMCTLFT3})
that provide information on the mesoscale DSF (\ref{Smesodef})
we first derive separately the asymptotic behaviour of the
Laplace transform and of the memory integral.

The mesoscale behaviour of (\ref{DSFepspt}) is given by the 
limit $\varepsilon\to 0$ of the auxiliary function (\ref{tspmeps}). 
To discuss this limit it is convenient to
introduce further auxiliary functions
\begin{eqnarray}
\tilde{s}^\ast(x,y) & := &
\int_{0}^\infty \rmd \tau \, \rme^{- \tau}
s ((x \tau)^{1/z_1}, (y \tau)^{1/z_2}) 
\label{tspm} \\
\check{s}^\ast(x,y) & := & \int_{0}^\infty \rmd \tau \, s ((x \tau)^{1/z_1}, (y \tau)^{1/z_2}) 
\label{cspm}
\end{eqnarray}
and the parameters
\begin{equation}
\mathring{\chi} := \min{(\chi_1,\chi_2)}, \quad 
\mathring{\theta} := \min{(\chi,\mathring{\chi})}.
\label{ringchidef}
\end{equation}

\begin{proposition}
\label{Prop:1}
For $\chi>0$
the limit
\begin{equation}
\tilde{s}^{\diamond}(\tilde{u}_1,\tilde{u}_2) 
:=  \lim_{\varepsilon \to 0} \varepsilon^{\mathring{\theta}-\chi} 
\tilde{s}_{\varepsilon}(\tilde{u}_1,\tilde{u}_2) 
\label{tsbulldef}
\end{equation}
is given by
\begin{equation}
\tilde{s}^{\diamond}(\tilde{u}_1,\tilde{u}_2)
= \left\{
\begin{array}{ll}
\begin{array}{ll}
1 &  \\
\end{array}  & \quad \chi < \mathring{\chi} \\[2mm]
\left\{
\begin{array}{ll}
\tilde{s}^\ast
(\tilde{u}_1,0) & \chi_1 < \chi_2\\
\tilde{s}^\ast
(\tilde{u}_1,\tilde{u}_2) & \chi_1 = \chi_2 \\
\tilde{s}^\ast
(0,\tilde{u}_2) & \chi_1 > \chi_2
\end{array} \right. & \quad \chi = \mathring{\chi}  \\[6mm]
\left\{
\begin{array}{ll}
\check{s}^\ast
(\tilde{u}_1,0) & \chi_1 < \chi_2 \\
\check{s}^\ast
(\tilde{u}_1,\tilde{u}_2) & \chi_1 = \chi_2 \\
\check{s}^\ast
(0,\tilde{u}_2) & \chi_1 > \chi_2 
\end{array} \right.  & \quad \chi > \mathring{\chi}.
\end{array} \right. 
\end{equation}
\end{proposition}

\begin{proof}
This is proved by noting that
\begin{eqnarray}
\varepsilon^{-\mathring{a}} \tilde{s}_{\varepsilon}(x,y) 
& = & \int_{0}^\infty \rmd \tau \, 
\rme^{- \varepsilon^{\mathring{a}}\tau} s 
((\varepsilon^{\chi_1-\chi+\mathring{a}} x \tau)^{1/z_1},
(\varepsilon^{\chi_2-\chi+\mathring{a}} y \tau)^{1/z_2}) \\
\varepsilon^{-\mathring{a}} \check{s}^\ast(x,y) 
& = &  \int_{0}^\infty \rmd \tau \, 
s ((\varepsilon^{\chi_1-\chi+\mathring{a}} u_1 \tau)^{1/z_1}, (\varepsilon^{\chi_2-\chi+\mathring{a}} u_2 \tau)^{1/z_2})
\end{eqnarray}
with an arbitrary substitution variable $\mathring{a}$ and that 
limit and integration can be interchanged by choosing 
$\mathring{a}=\chi-\mathring{\theta}$.
This yields
\begin{eqnarray}
\tilde{s}^{\diamond}(x,y) 
& = & 
\int_{0}^\infty \rmd \tau \, \lim_{\varepsilon \to 0} \rme^{- \varepsilon^{\chi-\mathring{\theta}}\tau}
s ((\varepsilon^{\chi_1-\mathring{\theta}} x \tau)^{1/z_1}, (\varepsilon^{\chi_2-\mathring{\theta}} y \tau)^{1/z_2}) .
\end{eqnarray}
The assertion then follows from the definition of $\mathring{\theta}$. \hfill \qed
\end{proof} 
Notice the independence of $\tilde{u}_2$ for $\chi_2>\chi_1$
and independence of $\tilde{u}_1$ for $\chi_2<\chi_1$
which means that $\tilde{s}^{\diamond}(p,\omega)$ is non-trivially
scale invariant
only for $\chi_2 \neq \chi_1$ and $\chi\geq\min{(\chi_1,\chi_2)}$. 

To analyse the memory integral we introduce the auxiliary functions
\begin{eqnarray}
\tilde{m}_{0}(x,y) & := & \frac{1}{2\pi} \int_{0}^{\infty} \rmd \tau \, \rme^{-\tau} 
\int_{- \infty}^{\infty} \rmd q \, 
s(q (x\tau)^{\frac{1}{z_1}}, q (y\tau)^{\frac{1}{z_2}})
\nonumber \\ & & \times
s (- q (x\tau)^{\frac{1}{z_1}}, - q (y\tau)^{\frac{1}{z_2}}) 
\label{tm0} \\
\tilde{m}^\ast(x,y) & := & \frac{1}{2\pi} 
\int_{0}^{\infty} \rmd \tau \, \rme^{-\tau} 
\int_{- \infty}^{\infty} \rmd q \, 
s(q (x\tau)^{\frac{1}{z_1}}, q (y\tau)^{\frac{1}{z_2}})
\nonumber \\ & & \times
s ((1- q) (x\tau)^{\frac{1}{z_1}}, (1- q) (y\tau)^{\frac{1}{z_2}}) 
\label{tmpm} \\
\check{m}^\ast(x,y) & := & \frac{1}{2\pi} 
\int_{0}^{\infty} \rmd \tau \, 
\int_{- \infty}^{\infty} \rmd q \, 
s(q (x\tau)^{\frac{1}{z_1}}, q (y\tau)^{\frac{1}{z_2}})
\nonumber \\ & & \times
s ((1- q) (x\tau)^{\frac{1}{z_1}}, (1- q) (y\tau)^{\frac{1}{z_2}}) 
\label{cmpm}
\end{eqnarray}
and the parameter
\begin{equation}
\mathring{\varphi} := \min{((\chi_1 - \mathring{\theta})/z_1,(\chi_2 - \mathring{\theta})/z_2)}.
\end{equation}

\begin{proposition}
\label{Prop:2}
For  $\chi>0$ and $\min{((\chi_1-\chi)/z_1,(\chi_2-\chi)/z_2)} < 1$ the limit
\begin{equation}
\tilde{m}^{\diamond}(\tilde{u}_1,\tilde{u}_2) 
:=  \lim_{\varepsilon \to 0} \varepsilon^{\mathring{\theta}-\chi+\mathring{\phi}} \tilde{m}_{\varepsilon}(\tilde{u}_1,\tilde{u}_2) 
\label{tmbulldef})
\end{equation}
is given by
\begin{equation}
\tilde{m}^{\diamond}(\tilde{u}_1,\tilde{u}_2)
= \left\{
\begin{array}{ll}
\left\{\begin{array}{ll}
\tilde{m}_{0}(\tilde{u}_1,0)
& \chi_1 < \chi_2 \\
\tilde{m}_{0}(\tilde{u}_1,\tilde{u}_2)
& \chi_1 = \chi_2 \\
\tilde{m}_{0}(0,\tilde{u}_2)
& \chi_1 > \chi_2  
\end{array}
\right. & \quad \chi < \mathring{\chi} 
\\[6mm]
\left\{
\begin{array}{ll}
\tilde{m}^\ast
(\tilde{u}_1,0) & \chi_1 < \chi_2\\
\tilde{m}^\ast
(\tilde{u}_1,\tilde{u}_2) & \chi_1 = \chi_2 \\
\tilde{m}^\ast
(0,\tilde{u}_2) & \chi_1 > \chi_2
\end{array} \right. & \quad \chi = \mathring{\chi}  \\[6mm]
\left\{
\begin{array}{ll}
\check{m}^\ast
(\tilde{u}_1,0) & \chi_1 < \chi_2 \\
\check{m}^\ast
(\tilde{u}_1,\tilde{u}_2) & \chi_1 = \chi_2 \\
\check{m}^\ast
(0,\tilde{u}_2) & \chi_1 > \chi_2 
\end{array} \right.  & \quad \chi > \mathring{\chi}.
\end{array} \right. 
\label{mbull}
\end{equation}
\end{proposition}

\begin{proof}

With arbitrary substitution variables $\mathring{a},\mathring{b}$
one has
\begin{eqnarray}
\tilde{m}_{\varepsilon}(x,y) 
& = & \varepsilon^{\mathring{a}-\mathring{b}} 
\frac{1}{2\pi} \int_{0}^{\infty} \rmd \tau \, 
\rme^{-\varepsilon^{\mathring{a}}\tau} 
\int_{-\varepsilon^{\mathring{b}-1}B(x,y)}^{\varepsilon^{\mathring{b}-1}B(x,y)} \rmd q \, 
\nonumber \\ & & 
s(
\varepsilon^{(\chi_1+\mathring{a}-\chi)/z_1-\mathring{b}}  
q (x \tau)^{1/z_1}, \varepsilon^{(\chi_2+\mathring{a}-\chi)/z_2-\mathring{b}}
q (y \tau)^{1/z_2})
\nonumber \\ & & 
\times s((1-\varepsilon^{-\mathring{b}}q) (\varepsilon^{\chi_1-\chi+\mathring{a}} x \tau)^{1/z_1},
(1-\varepsilon^{-\mathring{b}}q) (\varepsilon^{\chi_2-\chi+\mathring{a}} y \tau)^{1/z_2}).
\end{eqnarray}
Taking $\mathring{a} = \chi - \mathring{\theta}$ and
$\mathring{b} = 
\mathring{\varphi}$ yields $\mathring{\varphi} < 1$ and
allows for interchanging limit and integration. Hence
\begin{eqnarray}
\tilde{m}_{\varepsilon}(x,y) 
& = & \varepsilon^{\chi - \mathring{\theta} - \mathring{\varphi}} \frac{1}{2\pi} \int_{0}^{\infty} \rmd \tau \, \int_{-\infty}^{\infty} \rmd q \, \lim_{\varepsilon \to 0}
\rme^{-\varepsilon^{\chi - \mathring{\theta}}\tau} 
\nonumber \\ & & 
\times s(
\varepsilon^{(\chi_1 - \mathring{\theta})/z_1-\mathring{\phi}}  
q (x \tau)^{1/z_1}, 
\varepsilon^{(\chi_2 - \mathring{\theta})/z_2-\mathring{\phi}}
q (y \tau)^{1/z_2})
\nonumber \\ & & 
\times s((1-\varepsilon^{-\mathring{\phi}}q) (\varepsilon^{\chi_1-\mathring{\theta}} x \tau)^{1/z_1},
(1-\varepsilon^{-\mathring{\phi}}q) (\varepsilon^{\chi_2-\mathring{\theta}} y \tau)^{1/z_2}).
\label{tmpmeps2}
\end{eqnarray}
(i) For $\chi < \mathring{\chi}$ one has $\mathring{\theta}=\chi$
and $\mathring{\phi} = 
\min{((\chi_1 - \chi)/z_1,(\chi_2 - \chi)/z_2)}$ in the range $0<\mathring{\phi}<1$. A further substitution of
variables $q \to q |p|$ in (\ref{tmpmeps2}) yields (\ref{mbull})
for the range $\chi < \mathring{\chi}$.\\
(ii) For $\chi = \mathring{\chi}$ one has $\mathring{\theta}=\chi=\mathring{\chi}$
and $\mathring{\phi} = 0$ from which 
(\ref{mbull}) for $\chi = \mathring{\chi}$ follows.\\
(iii) For $\chi > \mathring{\chi}$ one has $\mathring{\theta} = 
\mathring{\chi}$ and $\mathring{\phi} = 0$. Making the
analogous substitution of variables in
(\ref{cmpm}) 
leads to (\ref{mbull})
for the range $\chi > \mathring{\chi}$.
\qed
\end{proof}

With the help of the 
functions $\tilde{s}^{\diamond}(\cdot,\cdot)$ 
and $\tilde{m}^{\diamond}(\cdot,\cdot)$ the notion of a
mesoscale solution of the mode coupling equation (\ref{dMCTLFT})
can now be made precise by making use of the representation
(\ref{dMCTLFT3}) of the discrete mode coupling equation.

\begin{definition}
\label{Def:mesosolFLT}
A generalized mesoscopic solution 
\begin{equation}
\tilde{S}^{\diamond}(p,\omega) := \lim_{\varepsilon \to 0} \varepsilon^{\mathring{\theta}} 
\tilde{S}(\varepsilon p, \varepsilon^{\chi} \omega)
= \frac{K}{\sqrt{2\pi}} \frac{1}{\omega} \tilde{s}^{\diamond}(\tilde{u}_1,\tilde{u}_2)
\label{tSmesodef}
\end{equation}
is a solution of the mode coupling equation (\ref{dMCTLFT}) with 
$q=\varepsilon p$ and $s = \varepsilon^{\chi} \omega$ in the limit 
$ \varepsilon\to 0$, such that
$\tilde{s}^{\diamond}(\tilde{u}_1,\tilde{u}_2)$
defined in (\ref{tsbulldef}) and $\tilde{m}^{\diamond}(\tilde{u}_1,\tilde{u}_2)$
defined in (\ref{tmbulldef}) solve the equation
\begin{equation}
\tilde{s}^{\diamond} 
= \lim_{\varepsilon \to 0}  \left[\varepsilon^{\chi -\mathring{\theta}}
+ \frac{\mathsf{e}(\varepsilon |p|)}{\varepsilon^2 |p|^2}
\left(\varepsilon^{2-\mathring{\theta}} \frac{|p|^2}{\omega} D  
+ 2 g^2  \varepsilon^{3+2\kappa-2\mathring{\theta}} \frac{|p|^3}{\omega^2}
\tilde{m}^{\diamond}
\right)\right]^{-1}
\label{dMCTLFT4}
\end{equation}
for some choice of the parameters $z_i,l_i,\chi_i$.
\end{definition}

To simplify notation we have suppressed in the defining equation 
(\ref{dMCTLFT4}) the dependence on the scaling
variables $\tilde{u}_1,\tilde{u}_2$ in the functions
$\tilde{s}^{\diamond},\tilde{m}^{\diamond}$
and in ratios of the coordinates $p,\omega$ which according to
(\ref{pzo}) are given by
\begin{equation}
\frac{|p|^2}{\omega} = 
\left(\frac{\tilde{u}_1}{l_1}
\right)^{\frac{z_2-2}{z_2-z_1}}
\left(\frac{\tilde{u}_2}{l_2} \right)^{\frac{2-z_1}{z_2-z_1}}, \quad
\frac{|p|^{3/2}}{\omega} = 
\left(\frac{\tilde{u}_1}{l_1}
\right)^{\frac{z_2-3/2}{z_2-z_1}}
\left(\frac{\tilde{u}_2}{l_2} \right)^{\frac{3/2-z_1}{z_2-z_1}}.
\label{pzo2}
\end{equation}
The generalized mesoscopic 
solution defined in (\ref{tSmesodef}) coincides with the ``natural''
mesoscale limit (\ref{Smesodef}) for $\chi\leq \mathring{\chi}$.
Generally one has
\begin{equation}
\tilde{S}^{\diamond}(p,\omega) = 
\left\{\begin{array}{ll}
\tilde{S}^\ast(p,\omega) & \chi\leq \mathring{\chi} \\
\lim_{\varepsilon \to 0} \varepsilon^{\mathring{\chi}} 
\tilde{S}(\varepsilon p, \varepsilon^{\chi} \omega) & \chi > \mathring{\chi}.
\end{array}\right.
\end{equation}
With these tools at hand we are in a position to state and prove the main results.

\begin{theorem}
\label{Theo:1}
Let $s (\cdot,\cdot)$ 
satisfy the integral equation
\begin{equation}
\tilde{s}^{\ast} (x,y)
=  \left[1+ y + x^2 \tilde{m}^{\ast}(x,y) \right]^{-1}, 
\quad x,y \in \R_0^+
\label{MCTint}
\end{equation}
with 
\begin{eqnarray}
\tilde{s}^\ast(x,y) & := &
\int_{0}^\infty \rmd \tau \, \rme^{- \tau}
s ((x \tau)^{1/z_1}, (y \tau)^{1/z_2}) 
\\
\tilde{m}^\ast(x,y) & := & \frac{1}{2\pi} 
\int_{0}^{\infty} \rmd \tau \, \rme^{-\tau} 
\int_{- \infty}^{\infty} \rmd q \, 
s(q (x\tau)^{\frac{1}{z_1}}, q (y\tau)^{\frac{1}{z_2}})
\nonumber \\ & & \times
s((1-q) (x\tau)^{\frac{1}{z_1}}, (1-q) (y\tau)^{\frac{1}{z_2}})
\end{eqnarray}
Then with the convention $z_2>z_1$ the mesoscopic limit
\begin{equation}
S^{\diamond}(p,\omega)  := 
\lim_{\varepsilon \to 0} \varepsilon^{\mathring{\theta}}
\tilde{S}(\varepsilon p,\varepsilon^{\chi}\omega) = \frac{K}{\sqrt{2\pi}} \frac{1}{\omega} 
\left\{\begin{array}{ll} 
1 & \chi < \mathring{\chi} \\[2mm]
\left\{\begin{array}{ll}
\tilde{s}^\ast(\tilde{u}_1,0) & \chi_1 < \chi_2  \\[2mm]
\tilde{s}^\ast(\tilde{u}_1,\tilde{u}_2) & \chi_1 = \chi_2  \\[2mm]
\tilde{s}^\ast(0,\tilde{u}_2) & \chi_1 > \chi_2 
\end{array} \right. 
\quad & \chi = \mathring{\chi} \\[8mm]
\left\{
\begin{array}{ll}
\check{s}^\ast(\tilde{u}_1,0) & \chi_1 < \chi_2  \\[2mm]
\check{s}^\ast(\tilde{u}_1,\tilde{u}_2) & \chi_1 = \chi_2  \\[2mm]
\check{s}^\ast(0,\tilde{u}_2) & \chi_1 > \chi_2 
\end{array} \right. 
& \chi > \mathring{\chi}
\end{array} \right.
\label{DSFtheo} 
\end{equation}
is a generalized mesoscopic solution of the mode coupling equation
(\ref{dMCTLFT})
if and only if 
\begin{equation}
\tilde{u}_1 \equiv \tilde{u}_1(p,\omega) :=  \frac{\sqrt{2} |g| |p|^{3/2}}{\omega}, \quad 
\tilde{u}_2 \equiv \tilde{u}_2(p,\omega) := \frac{D p^{2}}{\omega}
\label{u1u2theo1}
\end{equation}
and $\chi_1=3/2+\kappa$, $\chi_2 = 2$. 
\end{theorem}


\begin{proof}
Since $\lim_{\varepsilon \to 0} \mathsf{e}(\varepsilon |p|)/(\varepsilon^2 |p|^2)=1$ 
it needs to proved that
\begin{equation}
\tilde{s}^{\diamond}(\tilde{u}_1,\tilde{u}_2)
= \lim_{\varepsilon \to 0}  \left[\varepsilon^{\chi -\mathring{\theta}}
+ 
\varepsilon^{2-\mathring{\theta}} \frac{D |p|^2}{\omega}  
+ \varepsilon^{3+2\kappa-2\mathring{\theta}} \frac{2 g^2  |p|^3}{\omega^2}
\tilde{m}^{\diamond}(\tilde{u}_1,\tilde{u}_2)
\right]^{-1}
\label{dMCTLFT4b}
\end{equation}
with $\tilde{s}^{\diamond}(\tilde{u}_1,\tilde{u}_2)$ from Proposition \ref{Prop:1} and $\tilde{m}^{\diamond}(\tilde{u}_1,\tilde{u}_2)$ from Proposition \ref{Prop:2}.

From the integral equation
(\ref{MCTint}) one gets for $\chi\leq\mathring{\chi}$
\begin{equation}
\tilde{s}^{\diamond}(\tilde{u}_1,\tilde{u}_2)
=  \left\{
\begin{array}{ll}
\begin{array}{ll}
1 &  \\
\end{array}  & \quad \chi < \mathring{\chi} \\[2mm]
\left\{
\begin{array}{ll}
\left[1+ \tilde{u}_1^2 \tilde{m}^{\ast}(\tilde{u}_1,0) \right]^{-1} & \chi_1 < \chi_2\\
\left[1+ \tilde{u}_2 + \tilde{u}_1^2 \tilde{m}^{\ast}(\tilde{u}_1,\tilde{u}_2) \right]^{-1} & \chi_1 = \chi_2 \\
\left(1+ \tilde{u}_2\right)^{-1} & \chi_1 > \chi_2
\end{array} \right. & \quad \chi = \mathring{\chi} 
\end{array} \right. 
\label{tsb2}
\end{equation}
and therefore (\ref{pzo2}) implies the following.\\
\underline{(Aa) $\chi = \mathring{\chi}$, $\chi_1 < \chi_2$:}
\begin{equation}
1+ \tilde{u}_1^2 \tilde{m}^{\ast}(\tilde{u}_1,0) 
= \lim_{\varepsilon \to 0}  \left[1
+ \varepsilon^{2-\chi_1} \frac{|p|^2}{\omega} D  
+ 2 g^2  \varepsilon^{3+2\kappa-2\chi_1} \frac{|p|^3}{\omega^2}
\tilde{m}^{\ast}(\tilde{u}_1,0)
\right]
\label{dMCTLFT4Aa}
\end{equation}
Since the l.h.s. is a function only of $\tilde{u}_1$ a non-trivial
solution exists only with $z_1=3/2$, $l_1 = \sqrt{2} |g|$, and
$\chi_1 = 3/2+\kappa$ in the range $0\leq \kappa < 1/2$. 
The parameters $z_2$, $l_2$ and $\chi_2$ are undetermined
in this range.\\
\underline{(Ab) $\chi = \mathring{\chi}$, $\chi_1 > \chi_2$:}
\begin{equation}
1+ \tilde{u}_2
= \lim_{\varepsilon \to 0}  \left[1
+ \varepsilon^{2-\chi_2} \frac{|p|^2}{\omega} D  
+ 2 g^2  \varepsilon^{3+2\kappa-2\chi_2} \frac{|p|^3}{\omega^2}
\tilde{m}^{\ast}(0,\tilde{u}_2)
\right]
\label{dMCTLFT4Ab}
\end{equation}
Since the l.h.s. is a function only of $\tilde{u}_2$ a non-trivial
solution exists only with $z_2=2$, $l_2=D$, and $\chi_2 = 2$
in the range $\kappa > 1/2$. The parameters 
$z_1$, $l_1$ and $\chi_1$ are undetermined in this range.
\\
\underline{(Ac) $\chi = \mathring{\chi}$, $\chi_1 = \chi_2$:}
 Using the results from cases (Aa) and (Ab)
requires $\kappa=1/2$ with $\chi_1 = \chi_2 = 2$ and
the mode coupling equation reduces to
\begin{equation}
\tilde{s}^{\ast} (\tilde{u}_1,\tilde{u}_2)
=  \left[1 + \tilde{u}_2 + \tilde{u}_1^2
\tilde{m}^{\ast}(\tilde{u}_1,\tilde{u}_2)
\right]^{-1}
\end{equation}
\label{dMCTLFT4Ac}
which is true by assumption.\\

\noindent \underline{(B) $\chi < \mathring{\chi}$:}
\begin{equation}
1 = \lim_{\varepsilon \to 0}  \left(1
+ \varepsilon^{2-\chi} \frac{|p|^2}{\omega} D  
+ 2 g^2  \varepsilon^{3+2\kappa-2\chi} \frac{|p|^3}{\omega^2}
\tilde{m}^{\diamond}
\right)^{-1}
\label{dMCTLFT4B}
\end{equation}
In the range $0\leq\kappa \leq 1/2$ one has $\mathring{\chi}=3/2+\kappa \leq 2$
and in the range $\kappa > 1/2$ one has $\mathring{\chi}=2$.
Since $\chi < \mathring{\chi}$ the equation is satisfied for all $\kappa$.\\

\noindent \underline{(C) $\chi > \mathring{\chi}$:} 
The mode coupling equation reduces to
\begin{equation}
\tilde{s}^{\diamond} 
= \lim_{\varepsilon \to 0} 
\left(\varepsilon^{2-\mathring{\chi}} \frac{|p|^2}{\omega} D  
+ 2 g^2  \varepsilon^{3+2\kappa-2\mathring{\chi}} \frac{|p|^3}{\omega^2}
\tilde{m}^{\diamond}
\right)^{-1}
\label{dMCTLFT4C}
\end{equation}
On the other hand, the integral equation (\ref{MCTint}) implies by 
substitution of variables
\begin{equation}
x \tilde{u}_1^{-1} \tilde{s}^{\ast} (x,\tilde{u}_2/\tilde{u}_1)
=  \left[\tilde{u}_1 x^{-1}+ \tilde{u}_2 + \tilde{u}_1^2 x \tilde{u}_1^{-1} \tilde{m}^{\ast}(x,x\tilde{u}_2/\tilde{u}_1) \right]^{-1}, 
\quad x,\tilde{u}_1,\tilde{u}_2 \in \R^+.
\end{equation}
Taking the limit $x\to\infty$ yields the integral equation
\begin{equation}
\check{s}^{\ast}(x,y)
=  \left[y + x^2 \check{m}^{\ast}(x,y) \right]^{-1}
\label{checkMCT}
\end{equation}
so that from Proposition \ref{Prop:1} one gets
\begin{equation}
\tilde{s}^{\diamond}(\tilde{u}_1,\tilde{u}_2)
=  \left\{
\begin{array}{ll}
\left[\tilde{u}_1^2 \check{m}^{\ast}(\tilde{u}_1,0) \right]^{-1} & \chi_1 < \chi_2 \\
\left[\tilde{u}_2 + \tilde{u}_1^2 \check{m}^{\ast}(\tilde{u}_1,\tilde{u}_2) \right]^{-1} & \chi_1 = \chi_2 \\
\tilde{u}_2^{-1} & \chi_1 > \chi_2 .
\end{array} \right.  
\label{tsb2b}
\end{equation}
With
\begin{equation}
\tilde{m}^{\diamond}(\tilde{u}_1,\tilde{u}_2)
= \left\{
\begin{array}{ll}
\check{m}^\ast
(\tilde{u}_1,0) & \chi_1 < \chi_2 \\
\check{m}^\ast
(\tilde{u}_1,\tilde{u}_2) & \chi_1 = \chi_2 \\
\check{m}^\ast
(0,\tilde{u}_2) & \chi_1 > \chi_2 
\end{array} \right. 
\end{equation}
from Proposition \ref{Prop:2} the proof of (\ref{DSFtheo})
becomes analogous to the case B.
\qed
\end{proof}

\begin{remark}
It is clear from (\ref{dMCTLFT4Ac}) 
that the integral equation (\ref{MCTint}) is not
only a sufficient but also a necessary condition.
\end{remark}

\begin{remark}
From the scaling variables (\ref{u1u2theo1}) one reads off
\begin{equation}
l_1 = \sqrt{2} |g|, \quad l_2 = D
\label{l1l2theo1}
\end{equation}
for the scale parameters and
\begin{equation}
z_1 = \frac{3}{2}, \quad z_2 = 2
\label{z1z2theo1}
\end{equation}
for the dynamical exponents.
\end{remark}

\begin{corollary}
\label{Coro:1}
Let $s (\cdot,\cdot)$ 
satisfy the integral equation (\ref{MCTint}).
Then the mesoscopic limit
\begin{equation}
\tilde{S}^{\ast}(p,\omega) := 
\lim_{\varepsilon \to 0} \varepsilon^{\chi}
\tilde{S}(\varepsilon p,\varepsilon^{\chi}\omega) 
\end{equation}
of the Laplace transform of the DSF is
given by
\begin{equation}
\tilde{S}^{\ast}(p,\omega) = \frac{K}{\sqrt{2\pi}} \frac{1}{\omega} 
\left\{\begin{array}{ll} 
1 & \chi < \mathring{\chi} \\[2mm]
\left\{\begin{array}{ll}
\tilde{s}^\ast(\tilde{u}_1,0) & \chi_1 < \chi_2  \\[2mm]
\tilde{s}^\ast(\tilde{u}_1,\tilde{u}_2) & \chi_1 = \chi_2  \\[2mm]
\tilde{s}^\ast(0,\tilde{u}_2) & \chi_1 > \chi_2 
\end{array} \right. 
\quad & \chi = \mathring{\chi} \\[8mm]
0& \chi > \mathring{\chi}
\end{array} \right.
\label{DSFcoro1}
\end{equation}
with the scaling variables (\ref{u1u2theo1})
and $\chi_1=3/2+\kappa$, $\chi_2 = 2$. 
\end{corollary}

\begin{corollary}
\label{Coro:2}
Let $s (\cdot,\cdot)$ 
satisfy the integral equation (\ref{MCTint}).
Then for $\chi > \mathring{\chi}$ the time integral \begin{equation}
\check{S}^{\ast}(p) = 
\lim_{\varepsilon \to 0} \varepsilon^{\mathring{\chi}}
\int_0^\infty \rmd \tau \, 
S(\varepsilon p, \tau)
\end{equation}
of the DSF is given by
\begin{eqnarray}
\check{S}^{\ast}(p) 
& = & 
 \frac{K}{\sqrt{2\pi}} 
\left\{\begin{array}{ll} \frac{1}{\sqrt{2}|g| |p|^{3/2}} \check{s}^\ast(1,0) & \chi_1 < \chi_2  \\[2mm]
\check{s}^\ast(\sqrt{2}|g| |p|^{3/2}, D |p|^{2}) & \chi_1 = \chi_2  \\[2mm]
\frac{1}{D |p|^{2}}  \check{s}^\ast(0,1) & \chi_1 > \chi_2 
\end{array} \right. 
\label{DSFcoro2}
\end{eqnarray}
with the scaling variables (\ref{u1u2theo1})
and $\chi_1=3/2+\kappa$, $\chi_2 = 2$. 
\end{corollary}

Corollary 2 follows from 
$\check{S}^{\ast}(p) 
= \lim_{\varepsilon \to 0} \varepsilon^{\mathring{\theta}}
\tilde{S}_{\varepsilon}(p, 0) \\
= \lim_{\omega \to 0} S^{\diamond}(p,\omega) $
and the properties
\begin{eqnarray}
x \check{s}^\ast(x,0) & = & \check{s}^\ast(1,0) \\
\frac{1}{\omega} \check{s}^\ast(\tilde{u}_1,\tilde{u}_2) & = &  \int_{0}^\infty \rmd \tau \, s (|p| (l_1 \tau)^{1/z_1}, |p|
(l_2 \tau)^{1/z_2}) \\
y \check{s}^\ast(0,y) & = & \check{s}^\ast(0,1)
\end{eqnarray} 
which are immediate consequences of the definition
(\ref{cspm}).

\subsection{Scaling properties of the DSF}

With the critical observation scale
\begin{equation}
\chi^\ast = \min{(3/2+\kappa,2)}
\label{critchi}
\end{equation}
Corollary \ref{Coro:1} can be 
rephrased as
\begin{equation}
S^\ast(p,t) = \frac{K}{\sqrt{2\pi}} \left\{
\begin{array}{ll}
\begin{array}{ll}
1 &  \\
\end{array}  & \quad \chi < \chi^\ast \\[2mm]
\left\{
\begin{array}{ll}
s(0,u_2) & \kappa > 1/2 \\
s(u_1,u_2) & \kappa = 1/2 \\
s(u_1,0) & 0 \leq \kappa < 1/2
\end{array} \right. & \quad \chi = \chi^\ast  \\[6mm]
\begin{array}{ll}
0 &  \\
\end{array}  & \quad \chi > \chi^\ast.
\end{array} \right. 
\label{MCTmesosol}
\end{equation}

On mesoscopic scale the solution is scale invariant for the critical value $\chi=\chi^\ast$ 
except at the point $\kappa=1/2$ and has the superdiffusive
KPZ dynamical exponent 
$z=3/2$ for $0\leq\kappa<1/2$ and the diffusive EW critical exponent 
$z=2$ for $\kappa>1/2$. The scaling solution is, as expected for
dynamical critical behaviour, independent of the microscopic
parameters except that they determine the non-universal scale
parameters $l_i$. On the moderately asymmetric 
KPZ line $0\leq\kappa<1/2$ -- which
includes the point $\kappa=0$ of strong asymmetry --
the DSF is given by the solution of the integral equation
\begin{equation}
\tilde{s}^{\ast} (x,0)
=  \left[1+ x^2 \tilde{m}^{\ast}(x,0) \right]^{-1}, 
\quad x \in \R_0^+
\label{MCTintKPZ}
\end{equation}
and was studied numerically in \cite{Frey96}. Comparison
with the exact Praehofer-Spohn scaling function for the KPZ fixed point 
shows that the MCT solution is not exact, but numerically not very far
from the exact curve \cite{Prae04}. 
On the weakly asymmetric EW line $\kappa>1/2$ the MCT
integral equation reduces to the simple linear equation
\begin{equation}
\tilde{s}^{\ast} (0,y)
=  \left(1+ y \right)^{-1}, 
\quad x \in \R_0^+
\label{MCTintEW}
\end{equation}
solved by a Gaussian which represents the exact EW
critical DSF implied by the rigorous treatment in \cite{Gonc12}.
The crossover point $\kappa=1/2$ 
between the two universality classes coincides with the point where
the fluctuations are governed by the SBE \cite{Gonc12}.
The absence of scale invariance of the SBE is mirrored in the
absence of scale invariance of the MCT solution 
(\ref{MCTmesosol}) obtained above. 

The trivial subcritical behaviour simply means that on slow 
observations scales $\chi < \chi^\ast$ (short times) the
spreading of the initially strictly local correlation
is not noticeable. Only the ballistic motion of the density peak
with rescaled velocity $v=b(1-2\rho)$ is captured by the DSF.

On the other hand, on very fast observation scales 
$\chi > \chi^\ast$ (long times) the fluctuations are uncorrelated,
corresponding to complete relaxation of the initial perturbation
of the invariant measure. However, the time integral
of Corollary \ref{Coro:2} ``remembers'' the initial 
space-time correlations and results in spatial long-range
correlations of the integrated DSF. For $\kappa > 1/2$ this 
behaviour is reminiscent of the distribution of density fluctuations 
of the open SSEP which converge to a gaussian space-time
field that is delta correlated in time but with long-range correlations 
in space \cite{Bern23}.

It is also interesting to note
 the {\it macroscopic} limit (\ref{Smacrodef}) given by 
\begin{equation}
S^{\bullet}(p,t) = \frac{K}{\sqrt{2\pi}} \left\{
\begin{array}{ll}
\begin{array}{ll}
1 &  \\
\end{array}  & \quad \chi < \chi^\ast \\[2mm]
\left\{
\begin{array}{ll}
s(0,\varepsilon^{1-a/z_2} u_2) 
& \kappa > 1/2 \\
s(\varepsilon^{1-a/z_1} u_1,\varepsilon^{1-a/z_2} u_2) 
& \kappa = 1/2 \\
s(\varepsilon^{1-a/z_1} u_1,0) 
& 0 \leq \kappa < 1/2 
\end{array} \right. & \quad \chi = \chi^\ast  \\[6mm]
\begin{array}{ll}
0 &  \\
\end{array}  & \quad \chi > \chi^\ast
\end{array} \right. 
\label{MCTmacrosol}
\end{equation}
with the scaling variables $u_i = u_i(p,t)$. 
The macroscopic DSF is scale invariant for all 
$a>0$ with non-trivial scale invariance appearing in the critical regime $\chi = \chi^\ast$, {\it including} the point $\kappa=1/2$. At the 
macroscopic observation scale set by $a=z_1=3/2$ 
this yields
\begin{equation}
S^{\bullet}(p,t) = 
\frac{K}{\sqrt{2\pi}} \left\{
\begin{array}{ll}
1  & \kappa > 1/2 \\
s(u_1,0) & \kappa = 1/2 \\
s(u_1,0) & 0 \leq \kappa < 1/2
\end{array} \right. 
\end{equation}
which exhibits at the crossover point $\kappa=1/2$ the same KPZ-like
behaviour as the range $0\leq \kappa<1/2$ where the
WASEP fluctuations are believed to be governed by the
KPZ fixed point, see \cite{Mate21,Quas23} for recent discussions.
On the other hand, for long time
scales with $a=z_2=2$ the mesoscale MCT predicts
\begin{equation}
S^{\bullet}(p,t) = 
\frac{K}{\sqrt{2\pi}} \left\{
\begin{array}{ll}
s(0,u_2) & \kappa > 1/2 \\
0 & \kappa = 1/2 \\
0 & 0 \leq \kappa < 1/2.
\end{array} \right. 
\end{equation}
Thus at the macroscopic diffusive observation scale $a=2$ decorrelation has set in 
at and below the crossover
point $\kappa=1/2$, i.e., for moderate and critical bias.
(Fig. \ref{Fig:WASEPMCT}).

\begin{figure}[t]
\includegraphics[scale=.65]{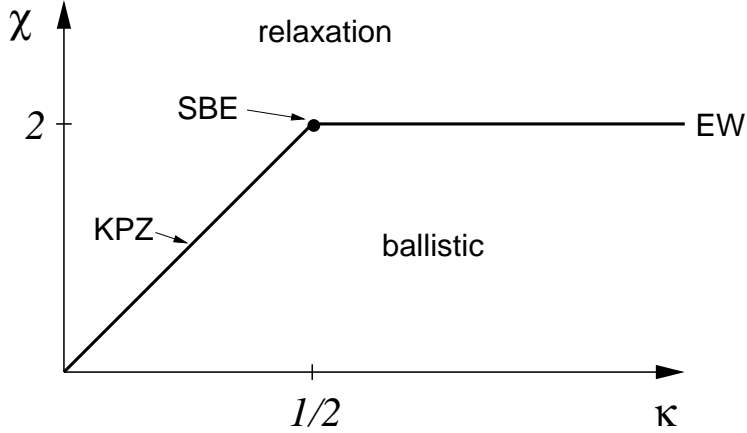}
%
%
\caption{Mesoscale MCT prediction of the phase diagram of the WASEP as a function of the 
bias exponent $\kappa$ and the observation scale
set by $\chi$. The critical line
(\ref{critchi})
corresponds to KPZ universality ($\kappa<1/2$), the
SBE point $\kappa=1/2$, and EW universality ($\kappa>1/2$).}
\label{Fig:WASEPMCT}       
\end{figure}

\section{Conclusions}

The main methodical novelties of the mesoscale scale approach to MCT 
are the absence of an {\it a priori} scaling assumption and the 
distinction between a mesoscopic and macroscopic space-time scale
(Def. \ref{Def:mesomacroscale}). Technically, this is facilitated by 
using the variables (\ref{zetadef}) which have simple scaling 
properties rather than the traditional microscopic space-time variables 
momentum $q$ and time $\tau$. As a result we obtain on mesoscopic 
scale for $\kappa=1/2$ a solution which is not
scale invariant but has features expected from the stochastic
Burgers equation which describes the fluctuations at this non-critical
crossover point between the (expected) KPZ
to the (established) EW universality class. For moderate
asymmetry with $\kappa<1/2$ MCT predicts the critical line
(\ref{critchi}) of non-trivial scale invariant behavior which is 
expected to be in the KPZ universality class, while for weak 
asymmetry $\kappa>1/2$ one gets EW criticality. 

There is no control over the approximation leading to the
discrete SBE and the mode coupling approximations (1) and (2). Only
the discrete gradient terms neglected in arriving at (\ref{apprx3})
are easily shown to lead to controlled subleading contributions to the 
large scale behaviour. It is therefore encouraging to see that
the line (\ref{critchi}) is consistent with an exact result, viz., 
the scaling $L^{-\chi^\ast}$ of the spectral gap of the generator 
of the WASEP with system size obtained from Bethe ansatz 
\cite{Kim95}. Moreover, for $\kappa >1/2$, MCT yields the 
EW universality proved 
mathematically rigorously in \cite{Gonc12}. Thus, even though the 
actual scaling form of the DSF does not match the exact 
Praehofer-Spohn scaling function in the KPZ regime, it appears that 
mesoscale MCT can predict correctly various non-trivial features of the 
DSF. In particular, it seems feasible to study in an approximate sense 
the crossover between KPZ and EW universality as well as diffusive 
corrections. The mesoscale approach also lends itself to generalization
to weakly nonlinear processes with more than one conserved quantity.

For using the mesoscale MCT to study processes other than the WASEP
it should be noted that
in Sec. 3 the microscopic mode coupling equations were derived 
directly from the Markov generator. This shortcut, bypassing nonlinear 
fluctuating hydrodynamics, is generally possible when the microscopic 
invariant measure is a product measure since in that case the
generator acting on the centered local random variables
produces the correct macroscopic current and therefore
the exact collective velocity and coupling strength that enter
the discrete SBE and consequently the mode coupling approximation. In the presence of correlations, however, one needs to
compute the exact current and appeal to nonlinear fluctuating 
hydrodynamics \cite{Spoh14} to postulate the correct
discrete stochastic partial differential equation on which MCT
is based.

\section*{Acknowledgements}
It is a pleasure to acknowledge stimulating discussions with
P. Gon\c{c}alves and M. Jara. This work is financially 
supported by FCT/Portugal through CAMGSD, IST-ID, Projects  UIDB/04459/2020 and UIDB/04459/2020 and by the FCT Grant  2022.09232.PTDC.

\end{document}